\documentclass[12pt]{amsart}

\usepackage{amsmath,amssymb, amscd}
\usepackage[dvipdfmx]{graphicx}

\makeatletter
\@addtoreset{equation}{section}
\makeatother
\usepackage{enumerate}
\usepackage{color}
\usepackage[all]{xy}

\usepackage{enumerate}

\def\ii{{\sqrt{-1}}}

\def\ee{\mathrm e}

\def\SL{\mathrm {SL}}

\def\CC{{\mathbb C}}
\def\ZZ{{\mathbb Z}}

\def\QQ{{\mathbb Q}}
\def\RR{{\mathbb R}}

\def\cJ{{\mathcal{J}}}

\def\cZ{{\mathcal{Z}}}
\def\cN{{\mathcal{N}}}
\def\cG{{\mathcal{G}}}
\def\cS{{\mathcal{S}}}

\def\cA{{\mathcal{A}}}
\def\cD{{\mathcal{D}}}

\def\PP{{\mathbb P}}

\def\hzeta{\hat{\zeta}}

\newtheorem{proposition}{Proposition}[section]

\newtheorem{remark}{Remark}[section]
\newtheorem{lemma}{Lemma}[section]

\def\book#1{\rm{#1}, }
\def\paper#1{\textit{#1}, }
\def\jour#1{\rm{#1}, }
\def\yr#1{({\rm{#1}) }}
\def\vol#1{\textbf{#1}}
\def\pages#1{\rm{#1}}

\def\by#1{{\rm{#1}, }}

\begin{document}

\title{Trigonal Toda lattice Equation}

\author{S.~Matsutani}

\begin{abstract}
In this article, we give the trigonal Toda lattice equation,
$$
-\frac{1}{2}\frac{d^3}{dt^3} q_{\ell}(t) 
= \ee^{q_{\ell+1}(t)}
+\ee^{q_{\ell+\zeta_3}(t)}
+\ee^{q_{\ell+\zeta_3^2}(t)}-3\ee^{q_\ell(t)},
$$
for a lattice point $\ell \in \mathbb{Z}[\zeta_3]$ as a 
directed 6-regular graph where $\zeta_3=\ee^{2\pi\ii/3}$, 
and its elliptic solution for the curve $y(y-s)=x^3$,
($s\neq 0$).
\end{abstract}


\maketitle

\section{Introduction}
The elliptic functions have high symmetries and generate many interesting
relations.
In the celebrated paper \cite{To}, Toda derived the Toda lattice equation
based on the addition formula of the elliptic functions.
Using the addition formulae of hyperelliptic curves \cite{EEMOP},
the hyperelliptic quasi-periodic solutions of the Toda lattice equation
are also obtained as in \cite{M, KdMP}. The derivation in \cite{M, KdMP}
can be regarded as
a natural generalization of Toda's original one.
The addition formulae for the Toda lattice equation are essential.

Recently  Eilbeck, Matsutani and \^Onishi introduced a new 
addition formula for the Weierstrass $\wp$ functions on an elliptic curve
$E$, $y(y-s)=x^3$, 
which is called the equiharmonic elliptic curve \cite{O}.
The curve $E$ has the automorphism, 
 the cyclic group action of order three as a Galois action \cite{EMO},
i.e., $\hzeta_3(x,y)=(\zeta_3x, y)$, 
where $\zeta_3=\ee^{2\pi\ii/3}$.

In this article, we use the new addition formula on $E$ in \cite{EMO}
and derive a non-linear differential and difference equation 
following the derivation in \cite{To,M, KdMP} 
as shown in Proposition \ref{prop:toda02}.
 Thus we call it {\it{the trigonal Toda lattice equation}}. 
The trigonal Toda lattice equation consists of
the third order differential and the trigonal difference operators,
which reflects the cyclic symmetry of the curve.
The difference operator agrees with the graph Laplacian 
of a directed 6-regular graph associated with Eisenstein integers 
$\ZZ[\zeta_3]$. 
The trigonal Toda lattice equation is defined over the infinite
directed 6-regular graph c.f. Proposition \ref{prop:toda03}.
It means that we provide the trigonal Toda lattice equation and
its elliptic function solution as a special solution.
Since the lattice given by the infinite
directed 6-regular graph appear in models in statistical
mechanics, e.g., \cite{B}, the new Toda lattice equation
 might show a nonlinear excitation of such models.

The contents are as follows.
In Section 2, we show the properties of the elliptic curve $E$.
We derive a new differential-difference equation, or
the trigonal Toda lattice equation, and its
elliptic function solution as an identity in the meromorphic
functions on $E$ in Section 3.
We give some discussions in Section 4.

\bigskip

\noindent
{\bf{Acknowledgments:}}
I would like to thank Yuji Kodama for helpful comments
and pointing out the Chazy equation, 
and Yoshihiro \^Onishi for valuable discussions. 
Further I am grateful to the two anonymous referees for 
their helpful comments and suggestions.

\section{Properties of the equiharmonic elliptic curve}
Let us consider an elliptic curve $E$ given by 
the affine equation
\begin{equation}
    y(y-s) = x^3,
\label{eq:yysx3}
\end{equation}
which is called the equiharmonic elliptic curve \cite{EMO, O}.
$E$ has the automorphism associated with the cyclic group of order three
as the Galois action on $E$; $\hzeta_3: E \to E$, 
$\hzeta_3((x,y))=(\zeta_3 x,  y)$ where $\zeta_3 = \ee^{2\pi\ii/3}$;
the action $\hzeta_3$ on $E$ is invariant.
We call it the trigonal cyclic symmetry.
The affine equation is expressed by
$$
\left(y-\frac{s}{2}\right)^2 = 
\left(x+\sqrt[3]{\frac{s^2}{4}}\right)
\left(x+\zeta_3\sqrt[3]{\frac{s^2}{4}}\right)
\left(x+\zeta_3^2\sqrt[3]{\frac{s^2}{4}}\right).
$$
By letting $\displaystyle{e_j=-\zeta_3^{1-j} \sqrt[3]{\frac{s^2}{4}}}$,
it corresponds to the Weierstrass standard form 
$$
(\wp')^2 = 4 \wp^3 -g_3=4(\wp-e_1)(\wp-e_2)(\wp-e_3),
$$
where $g_3 =-4s^2$ using the Weierstrass $\wp$-function,
$$
\wp(u)=x(u),\quad \wp'(u) = 2y(u)-s, \quad
y(u)=\frac{1}{2}(\wp'(u)+s),
$$
for the elliptic integral
$$
u = \int^{(x,y)}_\infty d u, \quad d u = \frac{d x}{2y -s }.
$$
It is known that since the image of the incomplete
elliptic integral agrees with the complex plane $\CC$,
the $\wp$-function (and thus $x$ and $y$) is expressed by
the Weierstrass sigma function,
$$
\wp(u) = -\frac{d^2}{d u^2} \log\sigma(u), \quad
\left(y(u) = -\frac{1}{2}
\left(\frac{d^3}{d u^3} \log \sigma(u) +s\right)\right).
$$
It means that 
$x(u)$ and $y(u)$ are considered as meromorphic functions on $\CC$.
The trigonal cyclic symmetry induces the action on $u$ and 
sigma function, i.e., for $u \in \CC$,
$$
\sigma(\zeta_3 u) = \zeta_3 \sigma(u), \quad
\wp(\zeta_3 u ) = \zeta_3 \wp(u), \quad
\wp'(\zeta_3 u) = \wp'(u).
$$
Eilbeck, Matsutani and \^Onishi
showed an addition formula of the elliptic sigma function of the curve $E$
\cite{EMO},
\begin{equation}
\frac{\sigma(u-v) \sigma(u-\zeta_3 v) \sigma(u-\zeta_3^2 v)}
{\sigma(u)^3 \sigma(v)^3 }= (y(u)-y(v)).
\label{eq:1-1}
\end{equation}
In this article, we consider the curve $E$ and this formula (\ref{eq:1-1}).

Let the elliptic integral from the infinity point $\infty$ to $(x,y)=(0, s)$
denoted by $\omega_s$, and similarly that to $(0,0)$ by $\omega_0$,
$$
    \omega_s = \int^{(0,s)}_\infty d u,\quad
   \omega_0 = \int^{(0,0)}_\infty d u, \quad d u =\frac{d x}{2y-s}.
$$
The complete elliptic integrals of the first and the second kinds are
given by
$$
\omega_i:= \int^{e_i}_\infty d u,\quad
\eta_i = \int_{\infty}^{e_i} x d u, \quad (i=1,2,3).
$$
Following Weierstrass' convention
$$
\omega' =\omega_1,\quad\omega''=\omega_3,\quad
\eta'=\eta_1,\quad\eta''=\eta_3,
$$
they satisfy the relations \cite{O, FKMPA}
$$
\omega''=\zeta_3\omega',\quad \eta''=\zeta_3^2\eta',
$$
$$
\omega_1+\omega_2+\omega_3=0, \quad
\eta_1+\eta_2+\eta_3 = 0, \quad
\eta'\omega''-\eta''\omega'=\frac{\pi\ii}{2}.
$$
Further for the branch points 
$(0,0)$ and $(s,0)$, 
the following relations hold:
\begin{lemma}\label{lm:2.2}
$$
x(\omega_0) = x(\omega_s) = 0,\quad
y(\omega_0) = 0,\quad
y(\omega_s) = s,
$$
$$
\omega_s =\frac{1-\zeta_3^2}{3}2\omega',\quad
\omega_0=-\omega_s, \quad
\omega'=\frac{3}{2}\frac{1}{\zeta_3-\zeta_3^2}
\frac{\Gamma\left(\frac{1}{3}\right)^2}
{s^{1/3}\Gamma\left(\frac{2}{3}\right)},\quad
\eta'=\frac{\pi\ii}{3\sqrt{3}}
\frac{s^{1/3}\Gamma\left(\frac{2}{3}\right)}
{\Gamma\left(\frac{1}{3}\right)^2}.
$$
\end{lemma}

\begin{proof} See \cite[Appendix C]{FKMPA}.\end{proof}

The image of the incomplete elliptic integrals is acted by
 $\SL(2, \ZZ)$ and the cyclic group $\zeta_3$.
For $u, v\in \CC$ $(v\neq0)$, we define a lattice
\begin{equation}
\cZ_{v,u}:=\ZZ[\zeta_3]v+u:=
\{\ell_1v +\ell_2 \zeta_3 v +u\ | \ \ell_1, \ell_2\in \ZZ\}.
\label{eq:Zv}
\end{equation}
Noting $\zeta_6 = \zeta_3+1$
for $\zeta_6:=\ee^{2\pi \ii/6}$,
 $\ZZ[\zeta_3]=\ZZ[\zeta_6]$.
Since $\cZ_{2\omega',0}$ agrees with the lattice of the periodicity, i.e.,
$x(u+L)=x(u)$, $y(u+L)=y(u)$ for $L \in \cZ_{2\omega',0}$,
the Jacobian $\cJ_E$ of the curve $E$ is given by 
$$
\cJ_E=\CC/ \cZ_{2\omega',0}, \quad
\kappa: \CC \to \cJ_E.
$$
These points of the integrals for the branch points of the curve $E$
are illustrated in Figure \ref{fig:omega}.
We also regard $x$ and $y$ as meromorphic functions on
 $\cJ_E$ due to their periodicity.

Further 
Lemma \ref{lm:2.2} shows that
$\omega_s$ and $\omega_0$ belong to $\frac{1}{3} \cZ_{2\omega',0}$.

\begin{figure}[ht]
\begin{center}
\includegraphics[width=0.65\hsize]{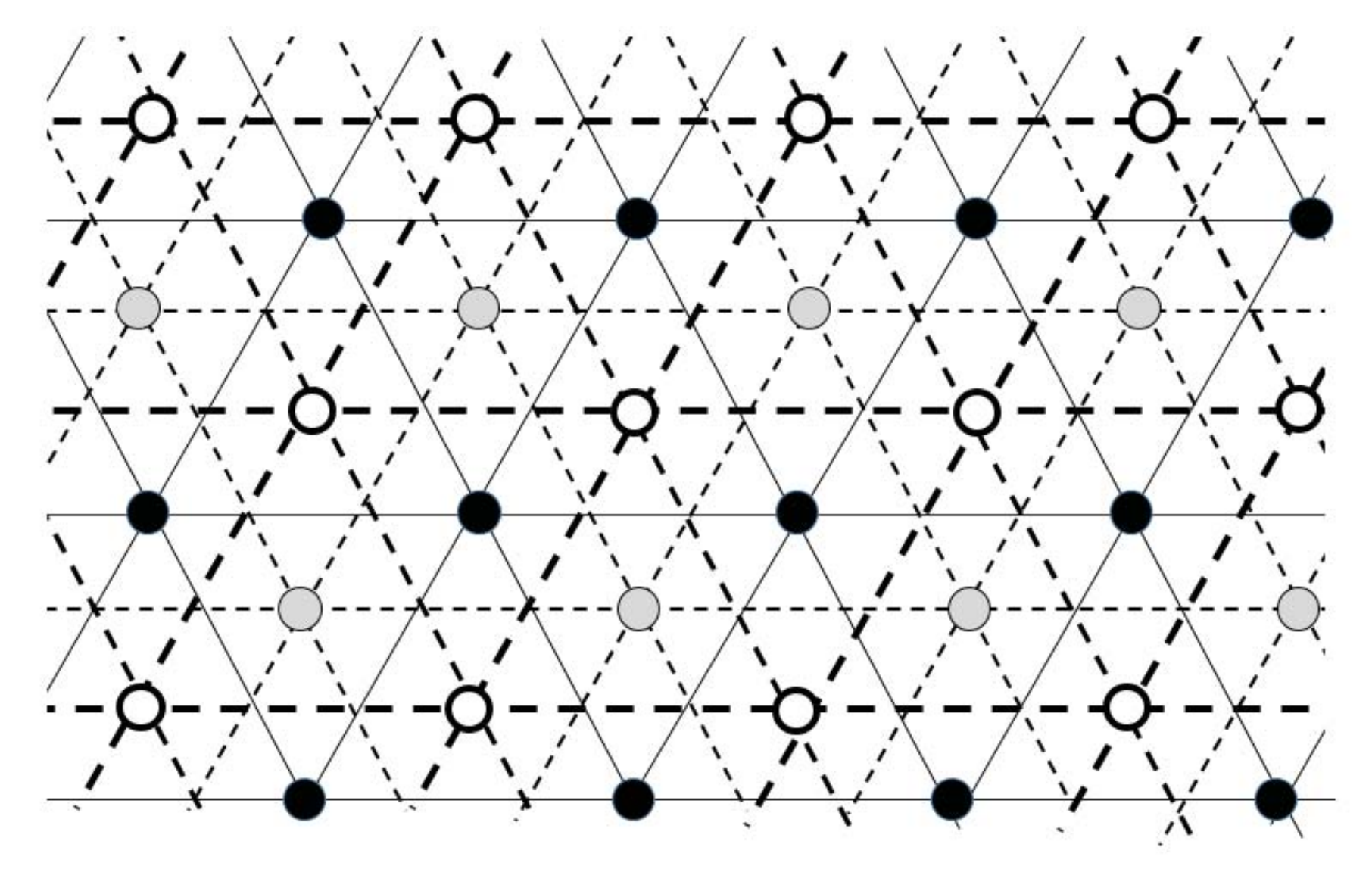}
\end{center}
\caption{
The lattice points of $\cZ_{2\omega',0}$:
The lattice points of $\cZ_{2\omega',0}$ are denoted by the black dots,
$\omega_0$ and $\omega_s$ with $\cZ_{2\omega',0}$ 
translations are denoted by
gray dots and white dots respectively.
}
\label{fig:omega}
\end{figure}

\section{The trigonal Toda lattice equation}

The addition formula (\ref{eq:1-1}) gives the following lemma:

\begin{lemma} \label{lmm:toda01}
The quantity
$$
q(u,v) := \log(y(u) - y(v)), \quad \left(y(u)-y(v) = \ee^{q(u,v)}\right)
$$
satisfies the relation
$$
-\frac{1}{2}\frac{d^3}{d u^3} q(u,v) = \ee^{q(u-v,v)}+\ee^{q(u-\zeta_3 v,v)}
+\ee^{q(u-\zeta_3^2 v,v)}-3\ee^{q(u,v)}.
$$
\end{lemma}

\begin{proof}
We consider the logarithm of both sides of (\ref{eq:1-1}) and differentiate
both side three times with respect to $u$. Then we obtain
$$
-\frac{d^3}{d u^3}\log(y(u) - y(v)) = 2 (y(u-v)+y(u-\zeta_3 v)
+y(u-\zeta_3^2 v)) - 6(y(u)).
$$
\end{proof}

\bigskip

The right hand side in Lemma \ref{lmm:toda01} 
should be expressed by a difference operator.
In order to express it, we prepare the geometry associated with
 Lemma \ref{lmm:toda01}.

We fix the complex numbers $u_0$ and $v_0$.
We regard $\cZ_{v_0,u_0}$ as the set of nodes 
$\cN_{v_0,u_0}$ of an infinite directed (oriented)
6-regular graph $\cG_{v_0,u_0}$ whose incoming degree and outgoing degree 
at each node are three, $\cN_{v_0,u_0}=\cZ_{v_0,u_0}$, 
as  $\cG_{v_0,u_0}$ is illustrated in Figure \ref{fig:dG} \cite{GR}.
Every node $n_{\ell}$ in $\cN_{v_0,u_0}$ is labeled by an 
Eisenstein integer $\ell\in \ZZ[\zeta_3]$.

\begin{figure}[ht]
\begin{center}
\includegraphics[width=0.45\hsize]{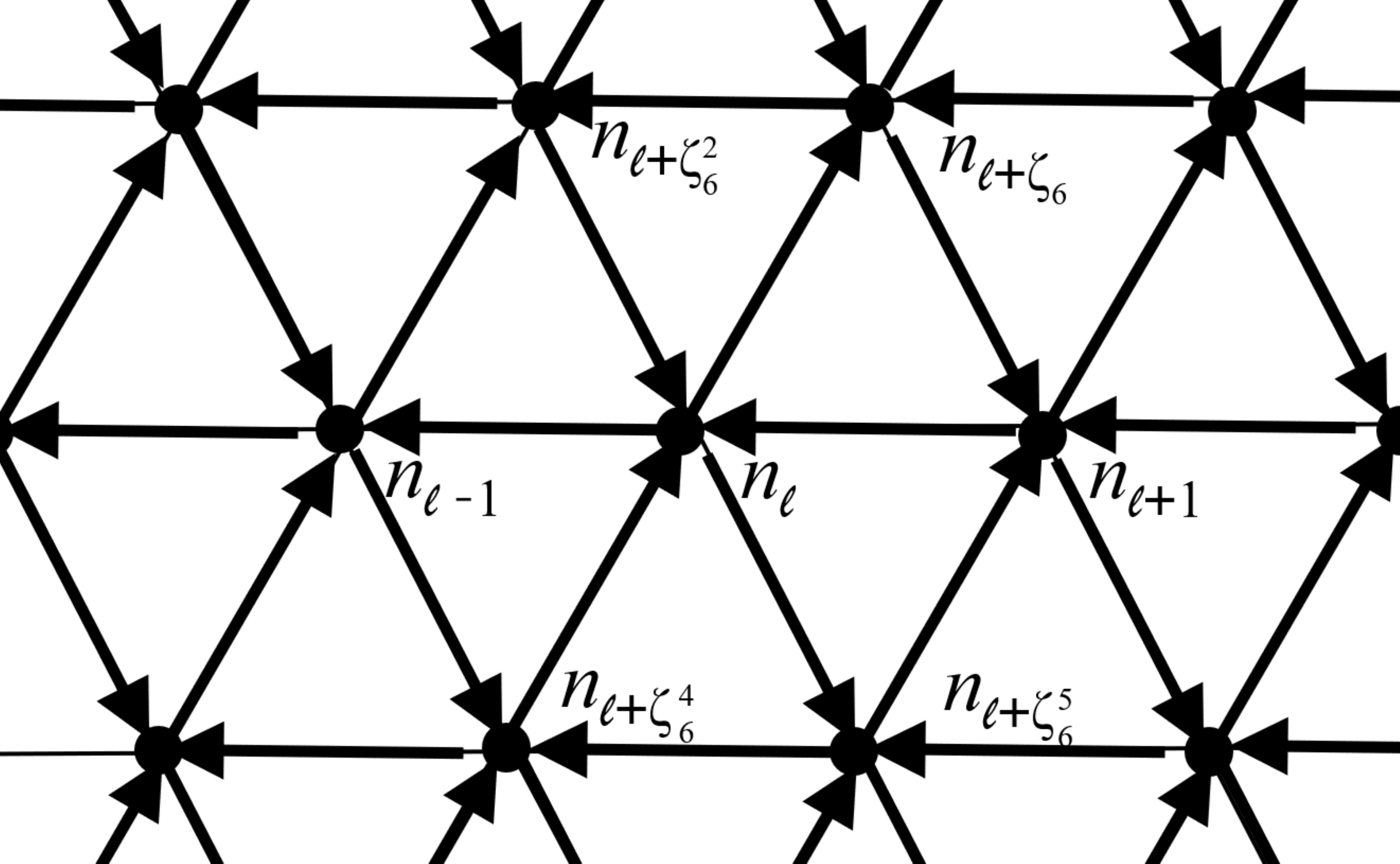}
\end{center}
\caption{
The directed graph $\cG_{v_0, u_0}$.
}
\label{fig:dG}
\end{figure}

It is noted that
 every $v_0 \in \CC$ is decomposed to $v_0=v_0'\omega' + v_0''\omega''$ 
using $v_0', v_0''\in \RR$.
Further the quotient set of the lattice points modulo $\cZ_{2\omega',0}$ 
is denoted by $\cN_{v_0, u_0}/\cZ_{2\omega',0}$ $=\kappa(\cN_{v_0, u_0})$.
The following are obvious:
\begin{lemma} \label{lemma:NZ}
\begin{enumerate}
\item
For $v_0 = v_0'\omega' + v_0''\omega''$ of $v_0', v_0''\in \QQ
\bigcap[0,2]$,
the cardinality $
|\cN_{v_0, u_0}/\cZ_{2\omega',0}| 
$ is finite for every $u_0 \in \CC$, and

\item
for $v_0 = v_0'\omega' + v_0''\omega''$ of $v_0', v_0''\in 
(\RR\setminus\QQ)\bigcap[0,2]$, $\cN_{v_0, u_0}/\cZ_{2\omega',0}$ is
dense in $\cJ_E$ for every $u_0 \in \CC$.
\end{enumerate}
\end{lemma}


Let us introduce the function spaces,
$\Omega$ and $\log\Omega$, 
$$
\Omega:=\{ Q: \CC \times \cN_{v_0, u_0} \to \PP\ |\  
\mbox{ meromorphic}\},\quad
\log\Omega:=\{q: \CC \times \cN_{v_0, u_0} \to \PP\ |\ 
\ee^q \in \Omega\}.
$$
For an Eisenstein integer $\ell \in \ZZ[\zeta_3]$ or 
$n_\ell\in\cN_{v_0, u_0}$, $t \in \CC$
and fixed $u_0, v_0\in \CC$,
 let us consider an element in $\log\Omega$,
\begin{equation}
q_{\ell}(t;u_0,v_0) := q(t+u_0-\ell v_0, v_0)=\log(y(t+u_0-\ell v_0) - y(v_0)),
\label{eq:q_ellvsy}
\end{equation}
which is denoted by $q_{\ell}(t)$ for brevity.

Lemma \ref{lmm:toda01} gives the nonlinear relation on $\log\Omega$:

\begin{proposition}\label{prop:toda02}
For $n_\ell\in\cN_{v_0, u_0}$, $t\in \CC$ and fixed $u_0, v_0\in \CC$,
$q_{\ell}(t)=q_{\ell}(t;u_0,v_0)$ satisfies the relation,
\begin{equation}
-\frac{1}{2}\frac{d^3}{dt^3} q_{\ell}(t) = \ee^{q_{\ell+1}(t)}
+\ee^{q_{\ell+\zeta_3}(t)}
+\ee^{q_{\ell+\zeta_3^2}(t)}-3\ee^{q_\ell(t)}.
\label{eq:prop3.1}
\end{equation}
\end{proposition}
It is emphasized that (\ref{eq:prop3.1}) 
can be regarded as  a 
differential-difference non-linear equation
and its special solution is given by (\ref{eq:q_ellvsy}) 
for the elliptic curve
(\ref{eq:yysx3}).
Its derivation is basically the same as Toda's original derivation
of Toda lattice equation \cite{To} and that in \cite{M, KdMP}. 
Further it is related to
the infinite graph $\cG_{v_0, u_0}$.
Thus we will call this relation {\it{trigonal Toda lattice equation}}.

We recall $\zeta_6=\zeta_3+1$.
For a given $n_\ell  \in \cN_{v_0,u_0}$,
let us consider subgraph $\cG_\ell \subset \cG_{v_0,u_0}$ 
given by its nodes $\cN_\ell:=\{n_\ell , n_{\ell +1}, n_{\ell +\zeta_6}, 
n_{\ell +\zeta_6^2}, \ldots,
n_{\ell+\zeta_6^5}\}$ $(\subset \cN_{v_0,u_0})$;
$\cN_\ell$ consists of the center point $n_{\ell}$
 with a hexagon $n_{\ell+\zeta_6^i}$
$(i=0, 1, \ldots, 5)$.
The submatrix of 
the incoming adjacency matrix $\cA_{\mathrm{in}}$ for $\cG_\ell$
is given by
$$
\cA_{\mathrm{in}}|_{\cG_\ell}=\begin{pmatrix}
0 & 1 & 0 & 1 & 0 & 1 & 0\\
0 & 0 & 1 & 0 & 0 & 0 & 1\\
1 & 0 & 0 & 0 & 0 & 0 & 0\\
0 & 0 & 1 & 0 & 1 & 0 & 0\\
1 & 0 & 0 & 0 & 0 & 0 & 0\\
0 & 0 & 0 & 0 & 1 & 0 & 1\\
1 & 0 & 0 & 0 & 0 & 0 & 0
\end{pmatrix}.
$$
The incoming degree matrix is given by the diagonal matrix
$\cD_{\mathrm{in}}$ whose diagonal
element is three. Thus we define the incoming Laplacian \cite{GR},
$$
\Delta_{\mathrm{in}}:= \cD_{\mathrm{in}} - \cA_{\mathrm{in}}.
$$
Let us consider the functions $q\in \log\Omega$ and $e^q\in \Omega$ 
 whose components
at $n_\ell$ are given by $q_\ell(t)$ and $\ee^{q_\ell(t)}$.
We regard them as column vectors for each $\ell \in \ZZ[\zeta_3]$.
Then the Laplacian acts on the vector spaces.

Using the incoming Laplacian,
 Proposition \ref{prop:toda02} is reduced to the following formula.
\begin{proposition}\label{prop:toda03}
Using the above notations, (\ref{eq:prop3.1}) is written by
$$
\frac{d^3}{d t^3} q(t) = -\Delta_{\mathrm{in}} e^q(t).
$$
\end{proposition}
It turns out that
the trigonal Toda lattice equation consists of the third order differential 
operators and trigonal graph Laplacian, which is a natural generalization of
the original Toda lattice equation \cite{To}, though it has not ever obtained
as far as we know.

As we obtain the equation, we will consider its solutions
(\ref{eq:q_ellvsy}), especially their initial condition $u_0$ and 
the configurations $\cG_{v_0,u_0}$:

\begin{remark}\label{rmk:21}
{\rm{
\begin{enumerate}
\item The domain of the solution $e^{q_\ell}$ of 
(\ref{eq:q_ellvsy}) is the Jacobian $\cJ_E$;
for $L \in \cZ_{2\omega',0}$, $e^{q_\ell}(t+L)=e^{q_\ell}(t)$.
From Lemma \ref{lm:2.2}, 
the periods $2\omega'$ and $2\omega''$ are scaled by $s^{-1/3}$.
Further in the projection $\pi:E \to \PP$, $(\pi((x,y))=y)$,
which determines the three special points $(0,s, \infty)$ in $\PP$,
the range of the solution $e^{q_\ell}$
 as a meromorphic function on $E$ is also parameterized by $s$
via $y$ and the governing equation (\ref{eq:yysx3}).
It is easy to find the $s$-dependence of $e^{q_\ell}$
and thus we may consider $s$ as a finite real number.

\item
For $v_0 (\neq 0)$ such that $v_0\neq \omega_0$,
$q(u,v_0)$ as a function with respect to $u$
diverges only at the points in $\cZ_{2\omega',0}$
and $\displaystyle{\bigcup_{i=0}^2(\zeta_3^i v_0 + \cZ_{2\omega',0})}$.
Their union is denoted by $\cS_{v_0}$.
It means that for an $\ell \in \ZZ[\zeta_3]$,
if the orbit of $q_\ell(t; u_0, v_0)$ in $t$
 avoids
 $\cS_{v_0}$, 
the value of $q_\ell(t; u_0, v_0)$ is finite.  

Let us consider its orbit whose value is finite value.
We restrict its domain 
$\CC \times \cZ_{v_0, u_0}$ to its real subspace 
$\RR\ee^{\alpha\ii} \times \cZ_{v_0, u_0}$ for 
a certain unit direction $\ee^{\alpha\ii}$ (i.e., $|\ee^{\alpha\ii}|=1$),

For the case (2) in Lemma \ref{lemma:NZ}, 
there are infinite many points at which 
$|q_\ell(t; u_0, v_0)|$ is greater than for every 
given positive number $1/\varepsilon$.
Thus we should employ the case (1) in Lemma \ref{lemma:NZ}.

\item
Let us assume that
$K:=|\cN_{v_0, u_0}/\cZ_{2\omega',0}|$ is finite.
For a certain direction $\ee^{\alpha\ii}$ in the 
complex plane and $u_0 \neq 0$,
we find the subspace $\ee^{\alpha\ii}\RR$ in $\CC$ such that
every $q_\ell(t_r \ee^{\alpha\ii}; u_0, v_0)$  does not diverge
for each $\ell \in \ZZ[\zeta_3]$
and $t\in \RR$, and satisfies
the trigonal Toda lattice equation,
\begin{equation}
\frac{d^3}{d t_r^3} q(t_r\ee^{\alpha\ii}) 
= -\ee^{3\alpha\ii}\Delta_{\mathrm{in}} e^q(t_r\ee^{\alpha\ii}).
\label{eq:Rmk:21}
\end{equation}
The conditions on $\ee^{\alpha\ii}$ and $u_0$ correspond to the conditions on
the embedding $\iota$ of $\RR^K$ into $\cJ_E$
such that the image of $\iota$ is compact and
disjoint from  $\cS_{v_0}/\cZ_{2\omega',0}$.
Under these conditions, we have the complex valued finite solutions of 
the trigonal Toda lattice equation (\ref{eq:Rmk:21}).

\end{enumerate}
}}
\end{remark}

An elliptic function solution of this equation is illustrated in 
Figure \ref{fig:esol} for 
 $v_0 = (1+\zeta_6)\omega'/13$, $u_0=v_0/2$,
$\ee^{\alpha\ii}=\omega'/|\omega'|$ and $s=1.0$.
\begin{figure}[ht]
\begin{center}

\hskip 0.1\hsize
\includegraphics[width=0.65\hsize]{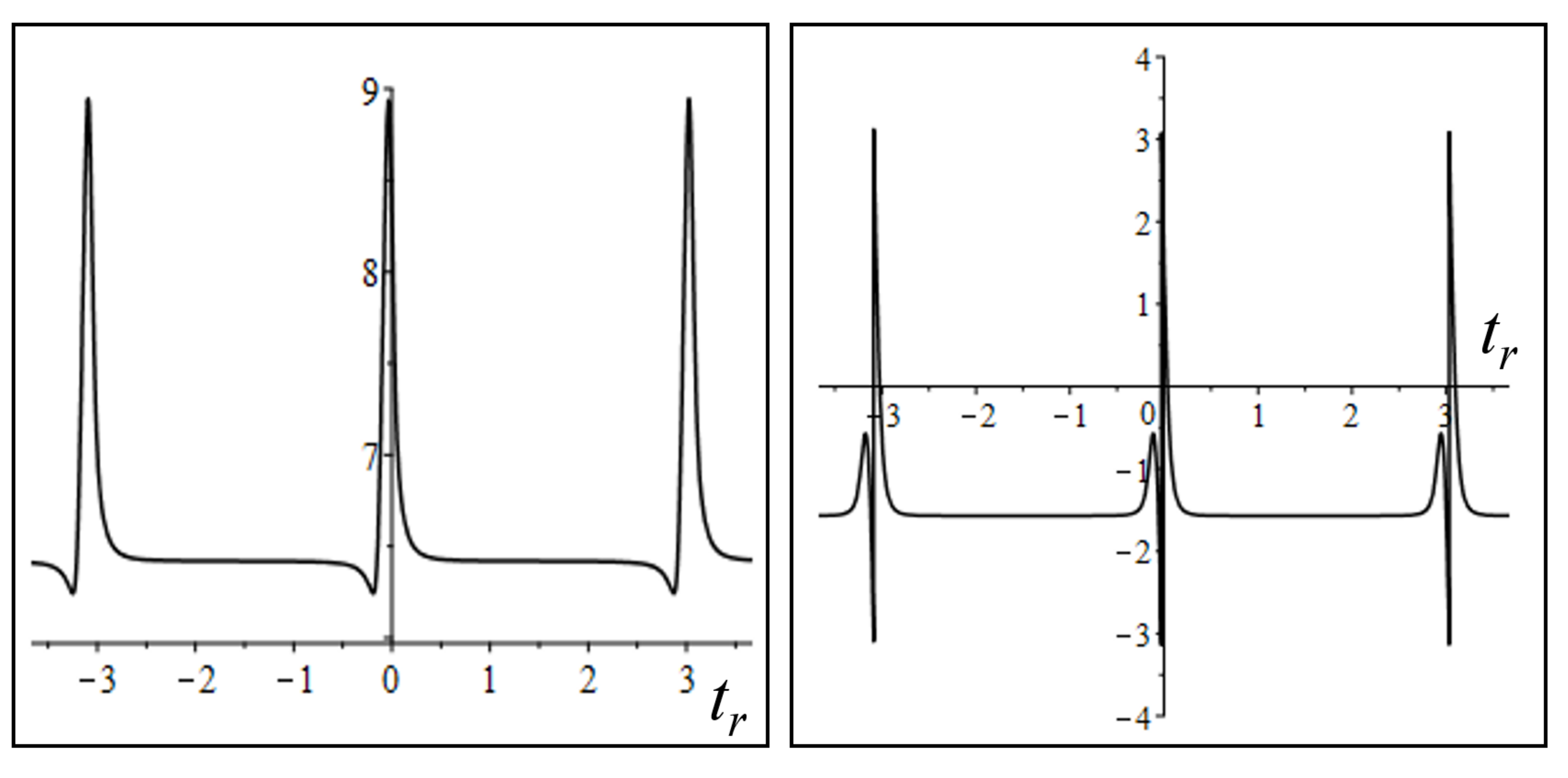}\newline

(a) \hskip 0.3\hsize (b)
\end{center}
\caption{
An elliptic solution of the trigonal Toda lattice equation
$q_0(t_r\ee^{\alpha\ii})$
 at $\ell=0$ for $v_0 = (1+\zeta_6)\omega'/13$, $u_0=v_0/2$,
$\ee^{\alpha\ii}=\omega'/|\omega'|$ and $s=1.0$, and $t_r\in \RR$:
(a) shows its real part whereas (b) corresponds to its imaginary part.}
\label{fig:esol}
\end{figure}

Let us consider the continuum limit of the the trigonal Toda lattice equation
as follows:
\begin{remark}\label{rmk:22}
{\rm{
\begin{enumerate}
\item It is noted that  $q(u,v_0)$  diverges for the limit $v_0 \to 0$ and thus
for this elliptic function solution $q_\ell(t)$,
we cannot obtain the continuum limit of the graph Laplacian 
$\Delta_{\mathrm{in}}$ and of the trigonal Toda lattice equation.

\item The elliptic curve $E$ becomes the three rational curves for 
the limit $s\to0$ as in \cite[Appendix C]{FKMPA}, and 
$y$ behaves like $\displaystyle{y=-\frac{1}{u^3}+\frac{1}{2} s^2 u^3 + o(s^3)}$
\cite{EMO}.
In the limit, 
the trigonal Toda lattice equation does not satisfy.
\end{enumerate}
}}
\end{remark}

\section{Discussion}

We derived the trigonal Toda lattice equation in
Propositions \ref{prop:toda02}
and \ref{prop:toda03}
based on the addition
formula (\ref{eq:1-1})
for the curve $E$ associated with the automorphism of the curve.
It is associated with the lattice, or the directed 6-regular
graph, given by the Eisenstein integers 
$\ZZ[\zeta_3]$.
It means that we have an nonlinear equation on the lattice and its
elliptic function solution.
Since there are  physical models based on the triangle lattice 
given by the infinite 6-regular graph
\cite{B},
this trigonal Toda lattice equation might describe a nonlinear excitation
in the models.

The third order differential equation reminds us of the Chazy equation,
which is a third order ordinary differential equation and posses 
Painlev\'e property.\cite{CO}.
However the trigonal
 Toda lattice equation cannot have a non-trivial continuum limit
because $E$ becomes the three rational curves for the limit $s\to0$ 
\cite{FKMPA} and $q_\ell(t; u_0, v_0)$ diverges  for the limit $v_0\to0$ as in 
Remark \ref{rmk:22}.
In other words, we could not directly 
argue the integrablity of the trigonal Toda 
lattice equation using the 
Chazy equation, even 
though both elliptic function solutions are closely related.
It means, in this stage, 
that it is not obvious whether the trigonal Toda lattice equation is 
an integrable equation as a time-development equation, 
and thus it is an open problem to determine the behavior of
its solution for every initial state as an initial value problem.

However the addition theorem in \cite[(A.3)]{EEMOP} for 
the genus three curve can be regarded as 
a generalization of the addition formula (\ref{eq:1-1}) for the cyclic action 
$\hzeta_3$ on curves. Thus it is expected that
the trigonal Toda lattice equation might have algebro-geometric solutions of
algebraic curves of higher genus.
Further this approach could be
generalized to more general curves, e.g., the genus three curve 
\cite{EEMOP} and more general curves with a trigonal cyclic group 
\cite{KmMP}.

\bigskip

\noindent
Shigeki Matsutani:\\
Faculty of Electrical, Information and Communication Engineering,\\
Institute of Science and Engineering\\
Kanazawa University,\\
Kakuma Kanazawa, 920-1192 JAPAN,\\
\texttt{s-matsutani@se.kanazawa-u.ac.jp}
\bigskip


\begin{thebibliography}{ABC DE}
\bibitem[B]{B}
\by{R.~J.~Baxter}
\book{Exactly Solved Models in Statistical Mechanics}
Academic Press, London, 1982.

\bibitem[CO]{CO}
 \by{P.~A.~Clarkson and P.~J.~Olver}
\paper{Symmetry and Chazy Equation} 
\jour{J. Diff. Equations}
\yr{1996} \vol{124} \pages{225-246}.

\bibitem[EEMOP]{EEMOP}
 \by{J.C. Eilbeck, V.Z. Enol'skii, S. Matsutani, Y. \^Onishi, and E. Previato}
\paper{Addition formulae over the Jacobian pre-image of hyperelliptic 
Wirtinger varieties} 
\jour{Journal f\"ur die reine und angewandte Mathematik (Crelles Journal)}
\yr{2008} \vol{2008} No. 619 \pages{37-48}

\bibitem[EMO]{EMO}
  \by{J. C. Eilbeck, S. Matsutani and Y. \^Onishi}
  \paper{Addition formulae for Abelian functions associated with
 specialized curves}
  \jour{Phil. Trans. R. Soc. A} \vol{369} \yr{2011} \pages{1245-1263}.

\bibitem[FKMPA]{FKMPA}
\by{Y. Fedorov, J. Komeda, S. Matsutani, E. Previato, and K. Aomoto}
\paper{The sigma function over a family of cyclic trigonal curves 
with a singular fiber} arXiv.1909.03858.

\bibitem[GR]{GR}
 \by{C. Godsil and G. Royle}
\book{Algebraic Graph Theory} 
Berlin, Springer, 2001.


\bibitem[KdMP]{KdMP}
  \by{Y. Kodama, S. Matsutani, and E. Previato}
\paper{Quasi-periodic and periodic solutions
of the Toda lattice via the hyperelliptic sigma function}
\jour{Annales de l'institut Fourier} \vol{63}
(2013) 655-688.

\bibitem[KmMP]{KmMP}
\by{J. Komeda, S. Matsutani, and E. Previato}
\paper{The sigma function for trigonal cyclic curves}
\jour{Lett. Math. Phys.}
\vol{109} (2019) 423-447.

\bibitem[M]{M}
\by{S. Matsutani}
\paper{Toda Equations and $\sigma$-functions of genera one and two}
\jour{J. Nonlinear Math. Phys.}
\vol{10} (2003) 555-569.


\bibitem[O]{O}
\by{Y. \^Onishi},
\paper{Complex multiplication formulae for 
hyperelliptic curves of genus three}
\jour{Tokyo J. Math.} \vol{21} \yr{1998} \pages{381--431}.


\bibitem[To]{To}
M. Toda,
\paper{Vibration of a Chain with Nonlinear Interaction}
\jour{J. Phys. Soc. Japan} \vol{22} (1967), 431--436.

\end{thebibliography}
\end{document}